\documentclass[11pt]{article}

\usepackage{cite,fullpage}
\usepackage{graphicx}
\usepackage{hyperref}

\usepackage{algorithm}
\usepackage{algorithmicx}
\usepackage[noend]{algpseudocode}
\usepackage{algpascal}

\usepackage{authblk}
\usepackage{booktabs}
\usepackage{subfigure}
\usepackage{multirow}
\usepackage{enumerate}

\usepackage{amsmath}
\usepackage{amssymb}
\usepackage{color}
\definecolor{red}{RGB}{255,0,0}
\definecolor{blue}{RGB}{0,0,255}
\definecolor{green}{RGB}{0,255,0}

\newenvironment{proof}{\noindent{\bf Proof:}}{$\square$ \vspace{3mm}}

\newcommand {\abs}[1]  {\left\vert#1\right\vert}
\newcommand {\set}[1]  {\left\{#1\right\}}

\newcommand {\defined} {\stackrel{def} {=}}

\newcommand {\runningtitle}[1] {\vspace{0.5ex}\noindent{\textbf{\boldmath #1:}}}

\newcommand{\ignore}[1] {}
\newcommand {\commentfig}[1] {#1}
\renewcommand{\commentfig}[1] {}

\newtheorem{theorem} {Theorem}
\newtheorem{lemma} {Lemma}
 
\newtheorem{claim}  {Claim}  
\newtheorem{corollary}  {Corollary}
\newtheorem{observation}  {Observation} 

\newtheorem{proposition} {Proposition}

\newcommand{\TT}  {{\cal T}}

\newcommand{\rep} {\left< T,\TT \right>}
\newcommand{\repprime} {\left< T',\TT' \right>}
\newcommand{\repbar} {\left< \bar{T},\bar{\TT} \right>}

\newcommand{\contract}[2] {{#1}_{/{#2}}}
\newcommand{\contracttte} {\contract{\TT}{e}}

\newcommand{\cg}{\textsc{ChordalGen}}
\newcommand{\gs}{\textsc{GrowingSubtree}}
\newcommand{\cn}{\textsc{ConnectingNodes}}
\newcommand{\pt}{\textsc{PrunedTree}}

\begin{document}

\title{The Complexity of Subtree Intersection Representation of Chordal Graphs and Linear Time Chordal Graph Generation\thanks{The first author acknowledges the support of the Turkish Academy of Science TUBA GEBIP award.}\thanks{A preliminary version of this work will be presented in the $SEA^2$ 2019 conference}}

\author[1]{T{\i}naz Ekim}
\author[2]{Mordechai Shalom}
\author[1]{Oylum \c{S}eker}

\affil[1]{
Department of Industrial Engineering, Bogazici University, Istanbul, Turkey
\footnote{tinaz.ekim,oylum.seker@boun.edu.tr}
}
\affil[2]{
TelHai Academic College, Upper Galilee, 12210, Israel
\footnote{cmshalom@telhai.ac.il}
}

\maketitle 

\begin{abstract}
It is known that any chordal graph on $n$ vertices can be represented as the intersection of $n$ subtrees in a tree on $n$ nodes \cite{Gavril74}.
This fact is recently used in \cite{seker17} to generate random chordal graphs on $n$ vertices by generating $n$ subtrees of a tree on $n$ nodes.
It follows that the space (and thus time) complexity of such an algorithm is at least the sum of the sizes of the generated subtrees assuming that a tree is given by a set of nodes. 
In \cite{seker17}, this complexity was mistakenly claimed to be linear in the number $m$ of edges of the generated chordal graph.
This error is corrected in \cite{seker2018generation} where the space complexity is shown to be $\Omega(m n^{1/4})$.
The exact complexity of the algorithm is left as an open question. 

In this paper, we show that the sum of the sizes of $n$ subtrees in a tree on $n$ nodes is $\Theta(m\sqrt{n})$. 
We also show that we can confine ourselves to contraction-minimal subtree intersection representations since they are sufficient to generate every chordal graph.
Furthermore, the sum of the sizes of the subtrees in such a representation is at most $2m+n$.
We use this result to derive the first linear time random chordal graph generator. 
Based on contraction-minimal representations, we also derive structural properties of chordal graphs related to their connectivity. 
In addition to these theoretical results, we conduct experiments to study the quality of the chordal graphs generated by our algorithm and compare them to those in the literature. 
Our experimental study indicates that the generated graphs do not have a restricted structure and the sizes of maximal cliques are distributed fairly over the range.
Furthermore, our algorithm is simple to implement and produces graphs with 10000 vertices and $4 . 10^7$ edges in less than one second on a laptop computer. 

\textbf{Keywords:}
Chordal graph, Representation Complexity, Graph Generation.
\end{abstract}

\section{Introduction}\label{sec:intro}
Chordal graphs are extensively studied in the literature from various aspects which are motivated by both theoretical and practical reasons. 
Chordal graphs have many application areas such as sparse matrix computations, database management, perfect phylogeny, VLSI, computer vision, knowledge based systems, and Bayesian networks (see e.g. \cite{Golumbic:2004:AGT:984029,IntersectionGT,pearl2014probabilistic,rose1972grap}).
Consequently, numerous exact / heuristic / parameterized algorithms have been developed for various optimization and enumeration problems on chordal graphs. 
The need for testing and comparing these algorithms motivated researchers to generate random chordal graphs  \cite{markenzon2008two,pemmaraju2005approximating,andreou2005generating}. 
A more systematic study of random chordal graph generators has been initiated more recently in \cite{seker17,seker2018generation}. 
The generic method developed in these papers is based on the characterization of chordal graphs as the intersection graph of subtrees of a tree \cite{Gavril74}, 
to which we will refer as a {\em subtree intersection representation}. 
In this method a chordal graph on $n$ vertices and $m$ edges is generated in three steps:\\
\begin{enumerate}
\item Generate a tree $T$ on $n$ nodes uniformly at random.\label{step:tree}
\item Generate $n$ non-empty subtrees $\set{T_1, \ldots, T_n}$ of $T$.\label{step:subtrees}
\item Return the intersection graph $G$ of $\set{V(T_1), \ldots, V(T_n)}$.\label{step:chordal}
\end{enumerate}

Three methods for generating subtrees in Step \ref{step:subtrees} have been suggested. 
In all these methods, every node of every subtree is generated. 
Steps \ref{step:tree} and \ref{step:chordal} being linear in the size of $G$, 
the time and space complexities of the algorithm are dominated by the sum of the sizes of the subtrees generated at Step \ref{step:subtrees} which is $\sum_{i=1}^n \abs{V(T_i)}$. 
In \cite{seker17}, this complexity was mistakenly claimed to be linear in the size of $G$, that is $O(n+m)$. 
In \cite{seker2018generation}, this mistake is corrected by showing that $\sum_{i=1}^n \abs{V(T_i)}$ is $\Omega(m n^{1/4})$,
leaving the upper bound as an open question. 
This question is crucial for the complexity of any chordal graph generator that produces every subtree intersection representation on a tree of $n$ nodes.
In this paper, we investigate the complexity of subtree intersection representations of chordal graphs.
We show that $\sum_{i=1}^n \abs{V(T_i)}$ is $\Theta(m\sqrt n)$. 
In other words, we both improve the lower bound of $\Omega(mn^{1/4})$ given in \cite{seker2018generation} and provide a matching upper bound. 
On the other hand, we show that the size of a "contraction-minimal" representation is linear, more precisely, at most $2m+n$. 
This result plays the key role in developing a linear time chordal graph generator, the first algorithm in the literature having this time complexity, to the best of our knowledge. 
Our algorithm is also simple to implement.
Our experiments indicate that it produces graphs for which the maximal clique sizes are distributed fairly over the range. 
Furthermore, the running time of the algorithm clearly outperforms the existing ones: 
graphs with 10000 vertices and $4.10^7$ edges are generated in less than one second on a personal computer.

Contraction-minimal representations of chordal graphs are further exploited to exhibit structural properties of chordal graphs. 
In particular, we show that every minimal representation of a chordal graph $G$ is on a tree with $t$ nodes where $t$ is the number of maximal cliques of $G$. Subsequently, we derive from this result that the connectivity of a chordal graph is at most $n-t$. Using this result, we show that our chordal graph generator can guarantee $k$-connectivity also in linear time. 

We proceed with definitions and preliminaries in Section \ref{sec:prelim}. 
Then, for technical reasons, we first consider contraction-minimal representations in Section \ref{sec:minimal}. 
We develop our linear time random chordal graph generator in Section \ref{subsec:linear}. 
We study the variety of chordal graphs generated by this algorithm in Section \ref{subsec:experiments}; 
to this end, following the practice in similar studies in the literature, we analyze maximal cliques of the generated graphs. 
In Section \ref{subsec:connectivity}, we analyze the links between contraction-minimal representations and the connectivity properties of chordal graphs they represent.
We proceed with the complexity of arbitrary subtree intersection representations in Section \ref{sec:general}. 
We conclude in Section \ref{sec:conclusion} by suggesting further research. 

\section{Preliminaries}\label{sec:prelim}
\runningtitle{Graphs}
We use standard terminology and notation for graphs, see for instance \cite{D12}.
We denote by $[n]$ the set of positive integers not larger than $n$. 
Given a simple undirected graph $G$, we denote by $V(G)$ the set of vertices of $G$ and by $E(G)$ the set of the edges of $G$.
We denote an edge between two vertices $u$ and $v$ as $uv$.
We say that 
a) the edge $uv \in E(G)$ is {\em incident} to $u$ and $v$, 
b) $u$ and $v$ are the endpoints of $uv$, and 
c) $u$ and $v$ are adjacent to each other.
We denote by $\contract{G}{e}$ the graph obtained from $G$ by contracting the edge $e$.
A \emph{chord} of a cycle $C$ of a graph $G$ is an edge of $G$ that connects two vertices that are non-adjacent in $C$.
A graph is \emph{chordal} if it contains no induced cycles of length 4 or more.
In other words, a graph is chordal if every cycle of length at least 4 contains a chord. 
A vertex $v$ of a graph $G$ is termed \emph{simplicial} if the subgraph of $G$ induced by $v$ and its neighbors is a complete graph. 
A graph $G$ on $n$ vertices is said to have a \emph{perfect elimination order} if there is an ordering $v_1,\ldots ,v_n$ of its vertices, 
such that $v_i$ is simplicial in the subgraph induced by the vertices $\set{v_i, \ldots ,v_{n}}$ for every $i \in [n]$. 
It is known that a graph is chordal if and only if it has a perfect elimination order \cite{Golumbic:2004:AGT:984029}. 

\runningtitle{Trees, subtrees and their intersection graphs}
Let $\TT = \set{ T_1, \ldots, T_n}$ be a set of subtrees of a tree $T$.
Let $G=(V,E)$ be a graph over the vertex set $\set{v_1, \ldots, v_n}$ where $v_i$ represents $T_i$ and such that $v_i$ and $v_j$ are adjacent if and only if $T_i$ and $T_j$ have a common node.
Then, $G$ is termed as the \emph{vertex-intersection graph} of $\rep$ and conversely $\rep$ is termed a subtree intersection representation, 
or simply a \emph{representation} of $G$. 
We will denote the intersection graph of $\rep$ simply as  $G(\TT)$. 
Gavril \cite{Gavril74} showed that a graph is chordal if and only if it is the vertex-intersection graph of subtrees of a tree. 
Throughout this work, we refer to the vertices of $T$ as \textit{nodes} to avoid possible confusion with the vertices of $G(\TT)$.

Let $G$ be a chordal graph with a representation $\rep$, and $j$ be a node of $T$. 
We denote as $\TT_j$ the set of subtrees in $\TT$ that contain the node $j$, i.e. $\TT_j \defined \set{T_i \in \TT | j \in V(T_i)}$.
Clearly, the set $\TT_j$ corresponds to a clique of $G$ that we will denote $V_j$.
It is also known that, conversely, if $K$ is a maximal clique of $G$ then $K = V_j$ for some node $j$ of $T$. 

Two sets of subtrees $\TT$ and $\TT'$ are \emph{equivalent} if $G(\TT) = G(\TT')$. 
Let $\TT$ be a set of subtrees of a tree $T$ and $e$ an edge of $T$. 
We denote by $\contracttte$ the set of subtrees of $\contract{T}{e}$ that is obtained by contracting the edge $e$ of every subtree in $\TT$ that contains $e$.
A set of subtrees is \emph{contraction-minimal} (or simply \emph{minimal}) if for every edge $e$ of $T$ we have $G(\contracttte) \neq G(\TT)$.

Throughout the rest of this work, $G$ is a chordal graph with vertex set $[n]$ and $m$ edges, $T$ is a tree on $t \leq n$ nodes and $\TT=\set{T_1, \ldots, T_n}$ is a set of subtrees of $T$ such that $G(\TT)=G$ and $\rep$ is contraction-minimal. We will adopt index $i\in [n]$ for vertices of $G$ and subtrees in $\TT$ whereas index $j\in [t]$ is used to denote the nodes of $T$. We also denote by $t_j \defined \abs{\TT_j}$ the number of subtrees in $\TT$ that contain the node $j$ of $T$.
The nodes of $T$ are numbered such that
\[
t_{t} = \max \set{t_j| j \in [t]}
\]
and all the other vertices are numbered according to a bottom-up order of $T$ where $t$ is the root.

In what follows, we first analyze contraction-minimal representations, then proceed into the analysis of the general case. 

\section{Contraction-minimal Representations}\label{sec:minimal}
We first show in Section \ref{subsec:linear} that the size of a contraction-minimal representation is at most $2m+n$. 
Based on this result, we derive a linear time algorithm to generate random chordal graphs. 
In Section \ref{subsec:experiments}, we conduct experiments to compare chordal graphs obtained by our algorithm to those in the literature. 
Our experimental study indicates that our method is faster than existing methods in the literature. 
Our algorithm produces graphs with 10000 vertices and $4.10^7$ edges in less than one second on a personal computer. 
In addition, the generated graphs do not have a restricted structure as far as the size of their maximal cliques are concerned. 
In Section \ref{subsec:connectivity}, we first show that minimal representations are exactly those representations of chordal graphs on a so-called clique tree of it. Subsequently, we analyze the  links between contraction-minimal representations and the connectivity of the corresponding chordal graphs. We show that the connectivity of a chordal graph is at most $n-t$ and derive from this result a modification of our random chordal graph generator which ensures $k$-connectivity of the produced chordal graph also in linear time.

\subsection{Chordal Graph Generation in Linear Time}\label{subsec:linear}
The following observation plays an important role in our proofs as well as in our chordal graph generation algorithm.
\begin{observation}\label{obs:noInclusions}
$\rep$ is minimal if and only if for every edge $jj'$ of $T$, none of $\TT_j$ and $\TT_{j'}$ contains the other 
(i.e., both $\TT_j \setminus \TT_{j'}$ and $\TT_{j'} \setminus \TT_j$ are non-empty).
\end{observation}
\begin{proof}
Suppose that $\TT_j \subseteq \TT_{j'}$ for some edge $jj'$ of $T$ and let $v$ be the node obtained by the contraction of the edge $jj'$. 
Then every pair of subtrees that intersect on $j$ also intersect on $j'$.
Thus, they intersect also on $v$ (after the contraction of $jj'$).
Conversely, every pair of subtrees that intersect on $v$ contains at least one of $j,j'$.
By our assumption, they contain $j'$, thus they intersect on $j'$.
Therefore, a pair of subtrees intersect in $\TT$ if and only if they intersect in $\contract{\TT}{jj'}$.
Therefore, $G(\contract{\TT}{jj'})=G(\TT)$, thus $\rep$ is not contraction-minimal.

Now suppose that $\TT_j' \setminus \TT_j \neq \emptyset$ and $\TT_j \setminus \TT_j' \neq \emptyset$ for every edge $jj'$ of $T$.
Then, for every edge $jj'$ of $T$ there exists a subtree that contains $j$ but not $j'$ and another subtree that contains $j'$ but not $j$. 
These two subtrees do not intersect, but they intersect on $v$ after the contraction of $jj'$.
Therefore, $G(\contract{\TT}{jj'}) \neq G(\TT)$ for every edge $jj'$ of $T$.
We conclude that $\rep$ is contraction-minimal.
\end{proof}

\begin{lemma}\label{lem:minrep}
Let $\rep$ be a minimal representation of some chordal graph $G$ on $n$ vertices and $m$ edges. 
There exist numbers $s_1, \ldots, s_t$ such that
\begin{eqnarray}
\forall j \in [t]~~s_j & \geq & 1, \label{eqn:sPositive}\\
\sum_{j=1}^{t} s_j & = & n,\label{eqn:sSumsToN}\\
s_j \leq t_j & \leq & \sum_{i=j}^{t}s_i \label{eqn:boundsForT}\\
2 \sum_{j=1}^{t} s_j t_j - \sum_{j=1}^{t} s_j^2 & = & 2m + n\label{eqn:graphSize}
\end{eqnarray}
\end{lemma}
\begin{proof}
Consider the following pruning procedure of $T$ that implies a perfect elimination order for $G$.
We first remove the leaf $j=1$ from $T$ and all the simplicial vertices of $G$ which are represented by the subtrees $T_i \in \TT$ that consist of the leaf $j=1$.
There is at least one such subtree by Observation \ref{obs:noInclusions}.
We continue in this way for every $j \in [t]$ until both $T$ and $G$ vanish.
Let $G^j$ be the remaining graph at step $j$ before node $j$ is removed, and $s_j$ be the number of simplicial vertices of $G^j$ eliminated with the removal of node $j$.
Clearly, the numbers $s_j$ satisfy relations (\ref{eqn:sPositive}) and (\ref{eqn:sSumsToN}).
Recall that $t_j = \abs{\TT_j}$.
To see that (\ref{eqn:boundsForT}) holds, observe that the number of subtrees removed at step $j$ is at most the number of subtrees containing node $j$.
Observe also that $t_j\leq n- \sum_{i=1}^{j-1}s_i$ as the subtrees eliminated at prior steps do not contain the node $j$ by the choice of the nodes to be removed at every step.

To show (\ref{eqn:graphSize}), let $e_j$ be the number of edges of $G$ that have been eliminated at phase $j$ of the pruning procedure (during which node $j$ of $\TT$ is removed). 
We recall that a clique of $s_j$ vertices is removed, each vertex of which is adjacent to $t_j-s_j$ other vertices of $G$. 
Therefore, $e_j = s_j (s_j-1)/2 + s_j (t_j-s_j)$, i.e., $2 e_j + s_j = 2 s_j t_j - s_j^2$. Summing up over all $j \in [t]$ we get
\[
2 \sum_{j=1}^{t} s_j t_j - \sum_{j=1}^{t} s_j^2 = \sum_{j=1}^{t} (2 e_j + s_j) = 2m + n.
\]
\end{proof}
We are now ready to prove the main result of this section.
\begin{theorem}\label{thm:MinimalUpperBound}
If $\rep$ is a minimal representation of some chordal graph $G$ on $n$ vertices and $m$ edges then
\[
\sum_{i=1}^n \abs{V(T_i)} \leq 2m + n.
\]
\end{theorem}
\begin{proof}
We first note that $\sum_{i=1}^n \abs{V(T_i)} = \sum_{j=1}^{t} t_j$ 
since every node of every subtree $T_i$ contributes one to both sides of the equation. We conclude as follows using Lemma \ref{lem:minrep}:
\[
\sum_{j=1}^{t} t_j \leq \sum_{j=1}^{t}s_j t_j \leq 2 \sum_{j=1}^{t}s_j t_j - \sum_{j=1}^{t}s_j^2 = 2m + n
\]
where the first inequality and the last equality hold by relations (\ref{eqn:sPositive}) and (\ref{eqn:graphSize}) of Lemma \ref{lem:minrep} and the second inequality is obtained by replacing  $t_j$ with $t_j + (t_j-s_j)$ and noting that $t_j-s_j \geq 0$. 
\end{proof}

\newcommand{\alg}{\textsc{GenerateContractionMinimal}}
\newcommand{\func}{\textsc{NewNode}}

We now present algorithm $\alg$ that generates a random contraction-minimal representation together with the corresponding chordal graph 
where every contraction-minimal representation has a positive probability to be returned.
It creates a tree $T$ at random by starting from a single node and every time adding a leaf $j'$ adjacent to some existing node $j$ that is chosen uniformly at random. 
Every time a node $j'$ is added, the algorithm performs the following: 
a) a non-empty set of subtrees consisting of only $j'$ is added to $\TT$, 
b) a proper subset of the subtrees in $\TT_j$ (i.e., those containing $j$) is chosen at random and every subtree of it is extended by adding the node $j'$ and the edge $jj'$, 
c) the graph $G$ is extended to reflect the changes in $\TT$. 
A pseudo code of the algorithm is given in Algorithm \ref{alg:ContractionMinimal}.

\alglanguage{pseudocode}
\begin{algorithm}[h]
\caption{$\alg$}\label{alg:ContractionMinimal}
\begin{algorithmic}[1]
\Require{$n \geq 1$}
\Ensure{A contraction-minimal representation $\rep$ with $\abs{\TT}=n$, and $G=G(\TT)$.}
\State $\TT \gets \emptyset$.
\State $G \gets (\emptyset, \emptyset)$.
\State $j' \gets \Call{\func}{}$.
\State $T \gets (\set{j'},\emptyset)$.
\While {$\abs{\TT} < n$}
\State Pick a node $j$ of $T$ uniformly at random. \label{lin:loopStart}
\State $j' \gets \Call{\func}{}$.
\State $T \gets (V(T) + j',~E(T) + jj')$.
\State Pick a proper subset $\TT'$ of $\TT_j$ at random.\label{lin:chooseProper}
\ForAll{$T_i \in \TT'$}
\State $T_i \gets (V(T_i) + j',~E(T_i) + jj')$.
\State $E(G) \gets E(G) \cup \set{i} \times V_{j'}$. 
\EndFor
\EndWhile
\Return $(\rep, G)$.
\Statex

\Function{\func}{}
\State Pick a number $k \in [n-\abs{\TT}]$ at random.\label{lin:PickANumber}
\State $j' \gets $ a new node.
\State $U \gets (\set{j'},\emptyset)$. \Comment{A tree with a single vertex}
\State $\TT \gets \TT \cup k$ copies of $U$.
\State $V_{j'} \gets $ a clique on $k$ vertices.
\State $G \gets G \cup V_{j'}$.
\State \Return $j'$.
\EndFunction

\end{algorithmic}
\end{algorithm}

\begin{theorem}\label{thm:linear}
Algorithm {\alg} generates a chordal graph in linear time. Moreover, it generates any chordal graph on $n$ vertices with strictly positive probability. 
\end{theorem}
\begin{proof}
The algorithm creates the tree $T$ incrementally, and the subtrees of $\TT$ are created and extended together with $T$. 
More precisely, the set of subtrees containing a node of $T$ is not altered after a newer node is created.
Note that, however, the subtrees themselves might be altered to contain newer nodes.
Consider an edge $jj'$ of $T$ where $j'$ is newer than $j$.
The sets $\TT_j$ and $\TT_{j'}$ of the subtrees containing $j$ and $j'$ respectively, are established by the end of the iteration that creates $j'$.
Since only a proper subset of $\TT_j$ is chosen to be extended to $j'$, at least one subtree in $\TT_j$ does not contain $j'$.
Furthermore, since $j'$ is created with a non-empty set of subtrees containing it, and none of these subtrees may contain $j$,
there is at least one subtree in $\TT_{j'}$ that does not contain $j$.
Therefore, by Observation \ref{obs:noInclusions}, we conclude that $\rep$ is contraction-minimal.

We proceed with the running time of the algorithm.
It is well known that the addition of a single node and the addition of a single edge to a graph can be done in constant time.
We observe that the number of operations performed by {\func} is $k + \binom{k}{2} = \abs{V_{j'}}+\abs{E(V_{j'})}$.
Therefore, the number of operations performed in all invocations of $\func$ is $n + \sum_{j' \in V(T)} \abs{E(V_{j'})}$.
Let $jj'$ be the edge added to $T$ at some iteration of $\alg$.
We now observe that the number of other operations (i.e., except the invocation of $\func$) performed during this iteration is 
exactly $\abs{E(G) \cap (V_j \times V_{j'})}$. We conclude that the number of operations of the algorithm is proportional to $n+\abs{E(G)}$.

Consider a contraction-minimal representation $\repbar$ with $\abs{\bar{\TT}}=n$ and let $\bar{t} = \abs{\bar{T}}$.
It remains to show that the algorithm returns $\repbar$ with a positive probability.

We show by induction on $h$, that at the beginning of iteration $h$ (the $h$'th time the algorithm executes Line \ref{lin:loopStart}) $\rep$ is a sub-representation of $\repbar$ with positive probability.
That is, $T$ is a subtree of $\bar{T}$ and $\TT$ consists of the non-empty intersections of the subtrees of $\bar{\TT}$ with $T$ (formally, $\TT = \set{\bar{T}_i[V(T)]~\vert~\bar{T}_i \in \bar{\TT}} \setminus \set{(\emptyset, \emptyset)}$) with positive probability. 
Let $\bar{j'}$ be a node of $\bar{T}$, and let $\bar{k} = \abs{\bar{\TT}_{\bar{j'}}}$. Clearly, $\bar{k} \leq n$. 
With probability $1/n$ the algorithm will start by creating a node with $\bar{k}$ trivial subtrees in which case $\rep$ is a sub-representation of $\repbar$.
Therefore, the claim holds for $h=1$.
Now suppose that $\rep$ is a sub-representation of $\repbar$ at the beginning of iteration $h$ with probability $p>0$.
If $T=\bar{T}$ then $\rep = \repbar$, thus $\abs{\TT}=\abs{\bar{\TT}}=n$ and the algorithm does not proceed to iteration $h$. 
Otherwise, $T$ is a proper subtree of $\bar{T}$, i.e. $\bar{T}$ contains an edge $jj'$ with $j \in V(T)$ and $j' \notin V(T)$.
At iteration $h$, $j$ will be chosen with probability $1/\abs{V(T)} \geq 1/n$ by the algorithm and the edge $jj'$ will be added to $T$, ensuring that $T$ is a subtree of $\bar{T}$
with probability at least $p/n$ at the end of iteration $h$.
The number $\bar{k}$ of subtrees in $\bar{\TT}$ that contain $j'$ but not $j$ is at most $n - \abs{\TT}$.
Since $\repbar$ is contraction-minimal, we have $\bar{k} \geq 1$.
Therefore, $\bar{k} \in [n-\abs{\TT}]$ and the algorithm creates $\bar{k}$ trivial subtrees in $j'$ with probability $1/(n - \abs{\TT}) > 1/n$.
Since $\repbar$ is contraction-minimal, the set of subtrees in $\bar{\TT}$ that contain both $j$ and $j'$ is a proper subset of $\bar{\TT}_j = \TT_j$. 
The algorithm chooses this proper subset with probability $1/(2^{\abs{\TT_j}}-1)$ and adds the edge $jj'$ to each of them.
We conclude that at the end of iteration $h$ (thus at the beginning of iteration $h+1$), 
$\rep$ is a sub-representation of $\repbar$ with probability at least $p/(n^2 2^n) > 0$.
\end{proof}

\subsection{Experimental studies}\label{subsec:experiments}

In this section, we present our experimental results to demonstrate the computational efficiency of $ \alg $ and 
to provide some insight into the distribution of chordal graphs it generates. 
We implemented the presented algorithm in C++, and executed it on a laptop computer with 2.00-GHz Intel Core i7 CPU.
The implementation of the algorithm spans only 70 lines of C++ code.
Our source code is available in \\
\url{http://github.com/cmshalom/ChordalGraphGeneration}.

Following the approach of works \cite{pemmaraju2005approximating,seker17,seker2018generation}, we consider the characteristics of the maximal cliques of the returned graph. 
Table \ref{tab:expresults} provides a summary of the computational results of our algorithm. 
The first column reports the number of vertices $n$. 
For every value of $n$, we use four different average edge density values of 0.01, 0.1, 0.5, and 0.8, 
where edge density is defined as $\rho=\frac{m}{n(n-1)/2}$ with $m$ being the number of edges in the graph. 
For each pair of values $n, \rho$, we performed ten independent runs and reported the average values across those ten runs.  
The table exhibits the number of connected components, the number of maximal cliques, and the minimum, maximum, and average size of the maximal cliques along with their standard deviation. 
The rightmost column shows the time in seconds that it takes the algorithm takes to construct one graph. 
In order to achieve the desired edge density values, we discarded the graphs that turned out to be outside the range $[(1-\epsilon)\rho,(1+\epsilon)\rho]$, for $\epsilon=0.05$. 
For $\rho \leq 0.1$, we adjusted the upper bound at Line \ref{lin:PickANumber} in function $\func$ so that graphs with small edge densities are obtained more probably.  

\begin{table}[htbp]
	\centering
	\caption{Experimental results of algorithm $\alg$}
	\scalebox{0.85}{
	\begin{tabular}{ccccccccc}
		\toprule
	
		\parbox[t]{1.7cm}{\centering $\boldsymbol{n}$}  &  \parbox[t]{1.7cm}{\centering \textbf{Density}} &  \parbox[t]{1.7cm}{\centering \textbf{\#\\ conn.\\comp.s}} & \parbox[t]{1.7cm}{\centering \textbf{\# \\maximal\\cliques}} & \parbox[t]{1.7cm}{\centering \textbf{Min \\clique \\size}} & \parbox[t]{1.7cm}{\centering \textbf{ Max\\ clique \\size}} & \parbox[t]{1.7cm}{\centering \textbf{Avg \\clique \\size}} & \parbox[t]{1.7cm}{\centering \textbf{Sd of\\clique \\sizes}} & \parbox[t]{1.7cm}{\centering \textbf{Time to\\build}}\\
		\midrule

        \multirow{4}[2]{*}{1000} & 0.010 & 24.6  & 422.2 & 1.0   & 17.2  & 5.8   & 2.9   & 0.002 \\
          & 0.100 & 1.0   & 62.8  & 7.3   & 125.0 & 43.3  & 26.8  & 0.003 \\
          & 0.500 & 1.0   & 9.5   & 92.5  & 520.0 & 255.7 & 127.1 & 0.008 \\
          & 0.780 & 1.0   & 6.5   & 167.2 & 767.4 & 393.1 & 201.0 & 0.010 \\
        \midrule
        \multirow{4}[2]{*}{2500} & 0.010 & 6.3   & 582.8 & 1.1   & 38.8  & 12.1  & 6.7   & 0.005 \\
          & 0.100 & 1.0   & 70.7  & 12.4  & 311.9 & 101.4 & 65.8  & 0.010 \\
          & 0.507 & 1.0   & 10.0  & 227.5 & 1357.7 & 617.6 & 330.4 & 0.030 \\
          & 0.808 & 1.0   & 7.2   & 379.9 & 1997.9 & 931.1 & 535.3 & 0.049 \\
        \midrule
        \multirow{4}[2]{*}{5000} & 0.010 & 1.7   & 703.6 & 1.8   & 75.7  & 21.8  & 12.9  & 0.014 \\
          & 0.102 & 1.0   & 78.5  & 22.7  & 635.1 & 196.1 & 131.0 & 0.035 \\
          & 0.503 & 1.0   & 9.7   & 457.6 & 2479.1 & 1236.1 & 636.2 & 0.166 \\
          & 0.796 & 1.0   & 7.2   & 775.5 & 3986.7 & 1914.0 & 1058.4 & 0.204 \\
        \midrule
        \multirow{4}[2]{*}{10000} & 0.010 & 1.4   & 825.0 & 2.8   & 146.4 & 40.4  & 25.3  & 0.036 \\
          & 0.100 & 1.0   & 90.0  & 32.0  & 1385.9 & 356.9 & 266.5 & 0.132 \\
          & 0.499 & 1.0   & 9.9   & 829.4 & 5343.8 & 2407.2 & 1366.3 & 0.659 \\
          & 0.797 & 1.0   & 9.2   & 902.7 & 8125.6 & 3065.4 & 2243.1 & 0.952 \\
        \bottomrule
	\end{tabular}%
}
	\label{tab:expresults}
\end{table}%

Algorithm $\alg$ produces connected chordal graphs for $\rho \geq 0.1 $. 
When the average edge density is 0.01, the average number of connected components decreases as $n$ increases. 
A minimum clique size of 1 for $\rho = 0.01$ and $n=1000$ implies that the disconnectedness of the graphs is due to the existence of isolated vertices. 
As for the running time of the algorithm, the linear time complexity shown in Theorem \ref{thm:linear} clearly manifests itself in the amount of time it takes to construct a chordal graph. 
The rightmost column of Table \ref{tab:expresults} shows that our algorithm constructs a chordal graph in less than one second on average, 
even when $n=10000$ and $\rho=0.797$, i.e., $m>4 \cdot 10^7$. 

We compare our results to those of the two other methods from the literature. 
The first one is algorithm {\cg} proposed in \c{S}eker \emph{et al.}'s work \cite{seker2018generation}, 
which is based on the subtree intersection representation of chordal graphs. 
This algorithm is presented along with three alternative subtree generation methods. 
Here, we only consider algorithm {\cg} together with the subtree generation method called {\gs}, 
because this one is claimed to stand out as compared to the other presented methods, as far as the distribution of maximal clique sizes are concerned. 
The second algorithm we compare to is Andreou \emph{et al.}'s algorithm \cite{andreou2005generating}. 
This algorithm is also used in \cite{seker2018generation} for comparison purposes, and we refer to the implementation therein. 
In order to obtain results comparable to those given in \cite{seker2018generation}, we use the same $n, \rho$ value pairs in our experiments.

We now compare the results in Table \ref{tab:expresults} to those reported in \cite{seker2018generation} for algorithm {\cg} with {\gs} and Andreou \emph{et al.}'s algorithm.
We observe that the number of maximal cliques of the graphs produced by $\alg$ is usually lower than the others,
and inevitably, their average clique sizes are higher than the others.
The most notable difference of our algorithm from the others is its running time. 
Whereas a running time analysis of Andreou \emph{et al.}'s algorithm has not been given in \cite{andreou2005generating}, 
the average running time of our implementation of their algorithm is of 477.1 seconds per generated graph, 
excluding graphs on 10000 vertices for which the algorithm was extremely slow. 
The average running times of our implementation of algorithm {\cg}  is 93.2, 4.7, 182.6 seconds with the subtree generation methods {\gs}, {\cn}, and {\pt}, respectively. 
Algorithm $\alg$, however, achieves an average running time of 0.14 seconds.

In our next set of experimental results, we investigate the distribution of the sizes of maximal cliques to get some visual insight into the structure of the chordal graphs produced. Figure \ref{fig:histograms_1000} shows the average number of maximal cliques across ten independent runs for $n=1000$ vertices and four edge density values.
The figure is comprised of three rows, each row describing the result of the experiments on one algorithm;
algorithms $\alg$, {\cg} combined with {\gs} method, and the implementation of Andreou \emph{et al.}'s algorithm \cite{andreou2005generating} as given in \cite{seker2018generation}. 
Every row consists of four histograms corresponding to four different average edge density values  $\rho = 0.01, 0.1, 0.5$, and $0.8$.
The bin width of the histograms is taken as five; that is, frequencies of maximal clique sizes are summed over intervals of width five (from one to five, six to ten, etc.)
and divided by the number of runs (i.e., ten) to obtain the average values. 
For a given $n$ and average edge density value, we keep the ranges of $ x $-axes the same in order to make the histograms comparable. 
The $y$-axes, however, have different ranges because maximum frequencies in histograms vary considerably. 

\begingroup

\begin{figure}[htbp]
\centering

\subfigure[Results from algorithm$\alg$.]
{\includegraphics[width=\textwidth]{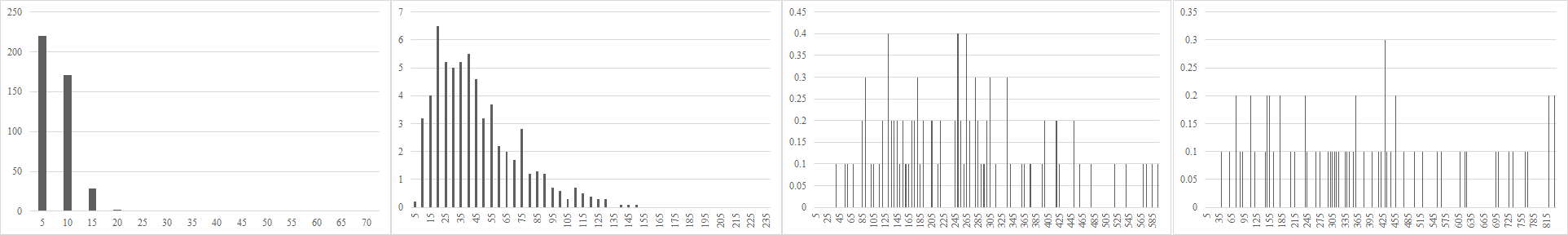}}
\label{fig:contractionMinimal_1000}

\subfigure[Results from algorithm {\cg} with {\gs} method.] 
{\includegraphics[width=\textwidth]{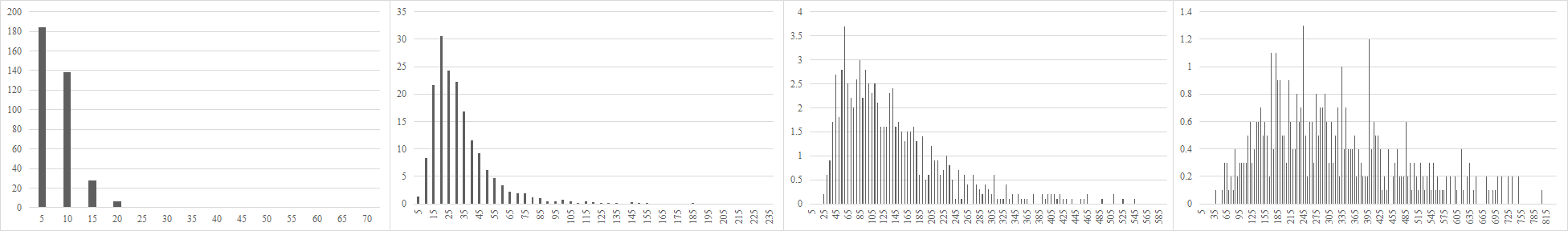}}
\label{fig:growingSubtree_1000}

\subfigure[Results from Andreou \emph{et al.}'s method \cite{andreou2005generating}.]
{\includegraphics[width=\textwidth]{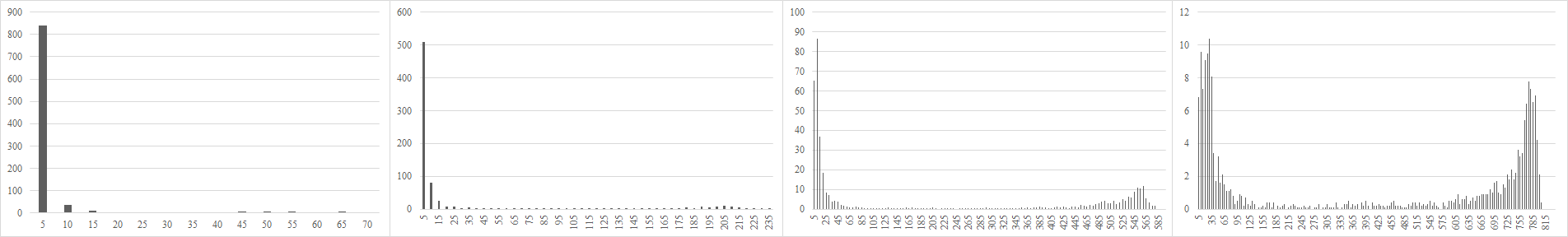}}
\label{fig:peoBased_1000}

\caption{Histograms of maximal clique sizes for $ n = 1000 $ and average edge densities 0.01, 0.1, 0.5, and 0.8 (from left to right).}
\label{fig:histograms_1000}
\end{figure}

\endgroup

The histograms in Figure \ref{fig:contractionMinimal_1000} reveal that the sizes of maximal cliques of graphs produced by our algorithm are not clustered around specific values; they are distributed fairly over the range. 
The shapes of the histograms for average edge densities 0.01 and 0.1 are similar for our algorithm and algorithm {\cg}, as we observe from the first half of Figure \ref{fig:contractionMinimal_1000} and \ref{fig:growingSubtree_1000}. 
For higher densities (as we proceed to the right), the sizes of maximal cliques are distributed more uniformly in the graphs generated by our algorithm; there is no obvious mode of the distribution. 
In the graphs produced by Andreou \emph{et al.}'s method, the vast majority of maximal cliques have up to 15 vertices when the average edge densities are 0.01 and 0.1. 
As we increase the edge density, frequencies of large-size maximal cliques become noticeable relative to the dominant frequencies of small-size maximal cliques. 
In any case, the range outside its extremes is barely used.

For brevity, we do not present the histograms for every $n$-value we consider in this study. 
Having presented the histograms for the smallest value of $n$ we consider, next we provide the set of results for a larger value of $n$. 
The implementation of Andreou \emph{et al.}'s algorithm \cite{andreou2005generating} turned out to be too slow to allow testing graphs on 10000 vertices in a reasonable amount of time. 
In order to present a complete comparison with the methods we look at from the literature, we provide the results for the next largest value of $n$ in Figure \ref{fig:histograms_5000}. 
From the histograms in Figure \ref{fig:histograms_5000}, we observe that the general distribution of maximal clique sizes do not change much with the increase in the number of vertices. 
Maximal clique sizes of chordal graphs produced by our algorithm are not confined to a limited area; they are distributed fairly over the range. 

To summarize, our experiments show that $\alg$ is by far faster than the existing methods in practice, in accordance with our theoretical bounds. 
Moreover, our inspection of the generated graphs in terms of their maximal cliques shows that the algorithm produces chordal graphs with no restricted structure. 

\begingroup

\begin{figure}[htbp]
\centering
\subfigure[Results from algorithm $\alg$.]
{\includegraphics[width=\textwidth]{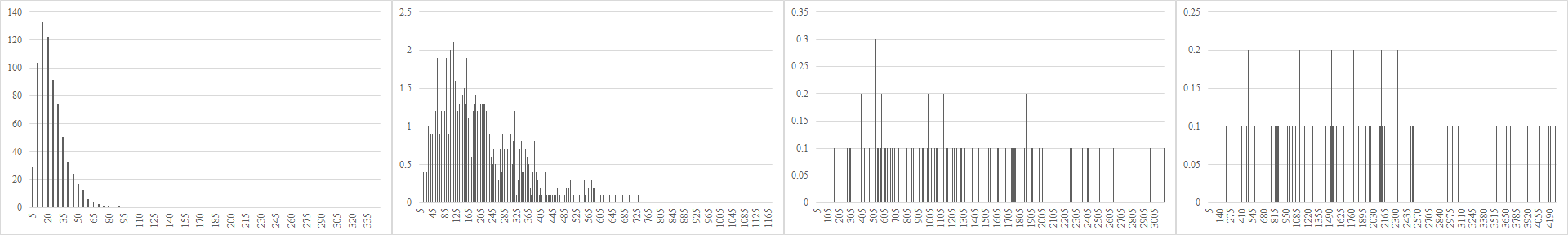}}
\label{fig:contractionMinimal_5000}

\subfigure[Results from algorithm {\cg} with {\gs} method.] 
{\includegraphics[width=\textwidth]{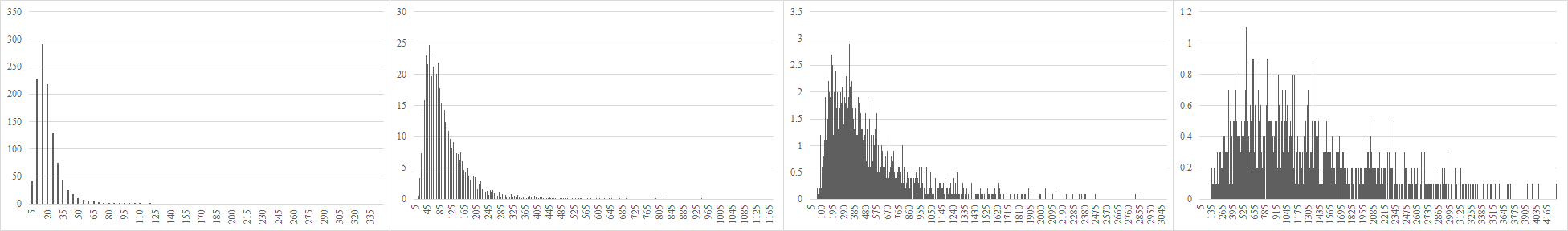}}
\label{fig:growingSubtree_5000}

\subfigure[Results from Andreou \emph{et al.}'s method \cite{andreou2005generating}.]
{\includegraphics[width=\textwidth]{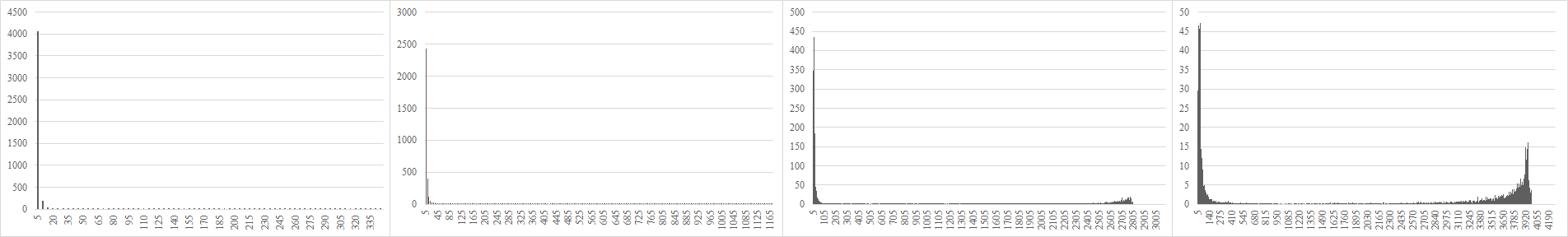}}
\label{fig:peoBased_5000}

\caption{Histograms of maximal clique sizes for $ n = 5000 $ and average edge densities 0.01, 0.1, 0.5, and 0.8 (from left to right).}
\label{fig:histograms_5000}

\end{figure}

\subsection{Structural Consequences}\label{subsec:connectivity}
By further analysis of contraction-minimal representations of chordal graphs, we show in this section that minimal representations are exactly those representations of chordal graphs on a so-called clique tree of it. Consequently, they provide us with important information about the connectivity properties of chordal graphs they represent.

\runningtitle{Maximal cliques and clique trees}
A \emph{clique tree} of a chordal graph $G$ is a tree $T$ such that a) there is a bijection between the maximal cliques of $G$ and the nodes of $T$ and
for every vertex $i \in V(G)$, the nodes of $T$ corresponding to maximal cliques of $G$ that contain $i$ form a subtree of $T$. 
It is well-known that a chordal graph on $n$ vertices has at most $n$ maximal cliques \cite{Golumbic:2004:AGT:984029}, 
thus every clique tree of it has at most $n$ nodes. 

It turns out that there is a one-to-one correspondence between minimal representations of a chordal graph and its clique trees as shown in the following.

\begin{proposition}\label{prop:mincliquetree}
$\rep$ is a minimal representation of a chordal graph $G$ if and only if $T$ is a clique tree of $G$ and every vertex $i$ of $G$ is represented by the subtree $T_i$ of $T$ induced by the nodes representing the maximal cliques of $G$ containing $i$.
\end{proposition}
\begin{proof}
Let $T$ be a clique tree of $G$ and $\TT = \set{T_1, \ldots T_n}$ where $T_i$ is the subtree of $T$ corresponding to the maximal cliques containing the vertex $i$ of $G$. We claim that $\rep$ is a minimal representation of $G$.
Indeed, consider an edge $e=jj'$ of $T$. 
Since $j$ and $j'$ correspond to distinct maximal cliques $V_j$ and $V_{j'}$ of $G$, neither $V_j  \setminus V_{j'}$ nor $V_{j'}  \setminus V_{j}$ is empty.
That is, there is a subtree in $\TT$ that contains $j$ but does not contain $j'$ and vice versa. 
By Observation \ref{obs:noInclusions}, $\rep$ is a minimal representation of $G$.

Let now $\rep$ be a minimal representation of $G$.
If $T$ consists of a single node, then $G$ is a complete graph and it has only one maximal clique that corresponds to the single node of $T$.
Therefore, $T$ is a clique tree of $G$.
Otherwise, let $e=jj'$ be an edge of $T$.
By Observation \ref{obs:noInclusions}, neither $\TT_j \setminus \TT_{j'}$ nor $\TT_{j'} \setminus \TT_j$ is empty,
i.e., neither $V_j \setminus V_{j'}$ nor $V_{j'} \setminus V_j$ is empty. 
Therefore both $V_j$ and $V_{j'}$ are maximal cliques of $G$.
Since $e$ is an arbitrary edge of $T$, we conclude that $V_j$ is a maximal clique of $G$ for every node $j$ of $T$.
Let now $V_j$ and $V_{j'}$ be two maximal cliques of $G$ corresponding to two (not necessarily adjacent) nodes $j$ and $j'$ in $T$, where $j \neq j'$. Let $j''$ be the node adjacent to $j$ on the path between $j$ and $j'$ in $T$ (with possibly $j''=j'$). By Observation \ref{obs:noInclusions}, there is a subtree $T_i \in \TT$ such that $T_i \in \TT_j$ and $T_i \notin \TT_{j''}$.
Therefore, $T_i \notin \TT_{j'}$, thus $V_j \neq V_{j'}$.
We conclude that the maximal cliques corresponding to nodes of $T$ are distinct.
Clearly, the subtree $T_i$ that represents a vertex $i$ of $G$ corresponds to the set of maximal cliques that contain $i$. 
Therefore, $T$ is a clique tree as claimed.
\end{proof}

As the number of maximal cliques is a graph invariant, Proposition \ref{prop:mincliquetree} implies that making any arbitrary representation of a chordal graph minimal by repetitively contracting edges always yields a minimal representation on a tree with the same number of nodes (which is the number of maximal cliques of $G$).

\begin{corollary}\label{cor:minimal}
Given a chordal graph $G$, every minimal representation $\rep$ of $G$ has a tree $T$ on $t$ nodes where $t$ is the number of maximal cliques of $G$. 
\end{corollary}

\runningtitle{Separators and minimal separators}
A set $S$ of vertices of a connected graph $G$ is called a \textit{separator} if $G \setminus S$ is not connected where $G\setminus S$ denotes the graph induced by $V(G) \setminus S$. 
For two non adjacent vertices $i$ and $i'$ of $G$, the set $S$ is an $i$-$i'$ separator if $i$ and $i'$ are in different connected components of $G\setminus S$. An $i$-$i'$ separator is \textit{minimal} if none
of its proper subsets separates $i$ and $i'$. 
We say that $S$ is a \emph{minimal separator} of $G$ if there exists two non adjacent vertices $i$ and $i'$ in $G$ such that $S$ is a minimal $i$-$i'$ separator. 
A graph is \textit{$k$-connected} if it has more than $k$ vertices and every separator of it has at least $k$ vertices. 
The (vertex) \textit{connectivity} $\kappa(G)$ of $G$ is the largest number $k$ such that $G$ is $k$-connected. 
Given a representation $\rep$ of a chordal graph $G$ and an edge $e$ of $T$, we denote by $\TT_e$ the set of subtrees in $\TT$ that contain $e$, 
and by $V_e$ the set of vertices of $G$ represented by them.

It is well-known that the minimal separators of a chordal graph are complete subgraphs \cite{Golumbic:2004:AGT:984029}. The following result describes more precisely the minimal separators of a chordal graph by means of its clique trees.

\begin{theorem}\label{thm:minsep}
\cite{HoLee89}
Given a chordal graph $G$ and any clique tree $T$ of $G$, a set of vertices $S$ is a minimal separator of $G$ if and only if $S=V_e$ for some edge $e$ of $T$. 
\end{theorem}
Proposition \ref{prop:mincliquetree} and Theorem \ref{thm:minsep} imply the following relation between minimal representations of a chordal graph and its minimal separators.
\begin{corollary}\label{cor:minsep}
Let $\rep$ be a minimal representation of a chordal graph $G$ where $\abs{V(T)}=t$.  Then
\begin{enumerate}[i)]
\item The set of minimal separators of $G$ is $\set{V_e | e \in E(T)}$.
\item The number of minimal separators of a chordal graph is at most $t-1$.
\item $\kappa(G) = \min \set{\abs{V_e} | e \in E(T)}$.
\end{enumerate}
\end{corollary}

\begin{proposition}
\label{prop:NumberOfTreesInEdge}
Let $\rep$ be a minimal representation of a chordal graph $G$ on $n$ vertices where $\abs{V(T)}=t$. Then,
\begin{enumerate}[i)]
    \item For every edge $e$ of $T$, the number of subtrees in $\TT$ that contain $e$ is at most $n - t$,
    \item $\kappa(G) \leq n - t$.
\end{enumerate}
\end{proposition}

\begin{proof}
\begin{enumerate}[i)]
\item Let $\rep$ be a minimal representation of a chordal graph $G$ on $n$ vertices and $e=jj'$ be an edge of $T$.
We prove by induction on $t$.
If $t=2$ then since $\rep$ is contraction-minimal, 
there is at least one subtree in $\TT_j$ that does not contain $jj'$ 
and at least one subtree in $\TT_{j'}$ that does not contain $jj'$, and we are done.

If $t>2$ at least one of $j,j'$ is not a leaf of $T$.
Therefore, $T$ contains at least one leaf $j'' \notin \set{j,j'}$.
Let $T'=T-j''$ and $\repprime$ be the sub-representation of $\rep$ where $\TT'$ consists of the non-empty intersections of the subtrees of $\TT$ with $T'$. Then $\repprime$ is a minimal representation of $G[V']$ where $V'$ is the set of vertices of $G$ represented by the subtrees in $\TT'$. By the induction hypothesis, 
$\TT'$ contains at least $\abs{V(T')}=t-1$ subtrees not containing $jj'$.
Furthermore, $\TT$ contains at least one subtree $T_i$ that contains $j''$ but not its neighbour in $T$.
Thus $T_i\in \TT \setminus \TT'$ and it does not contain $jj'$.
Therefore, $\TT$ contains at least $\abs{V(T')}+1=t$ subtrees that do not contain $jj'$.
We conclude that $\TT$ contains at most $n - t$ subtrees that contain $e$.

\item We have $\abs{V_e} = \abs{\TT_e} \leq n-t$. The result follows from Corollary \ref{cor:minsep}.
\end{enumerate}

\end{proof}

\begin{corollary}\label{cor:AlgModification}
With the following modifications, $\alg$ generates $k$-connected chordal graphs in linear time.
\begin{enumerate}[i)]
    \item At line \ref{lin:chooseProper} where a proper subset to contain the edge $jj'$ is chosen, pick a proper subset $\TT'$ of $\TT_j$ of cardinality at least $k$.
    \item At the first invocation of {\func}, create a clique of at least $k+1$ vertices.
\end{enumerate}
\end{corollary}

\begin{proof}
By Corollary \ref{cor:minsep}, the choice of a proper subset of cardinality at least $k$, i.e., 
the successful execution of (the modified)  Line \ref{lin:chooseProper} at every node is necessary and sufficient for $k$-connectivity.
Suppose that the execution of Line \ref{lin:chooseProper}  fails at some node of $T$, and let $j$ be the first node that this happens.
If $j$ is not the first node of $T$, let $j''$ be the node chosen by the algorithm when $j$ is created.
Since Line \ref{lin:chooseProper} did not fail at $j''$, at least $k$ subtrees from $\TT_{j''}$ contain $j''j$.
Moreover, at least one more subtree is created by {\func} when $j$ is created.
Therefore, $\TT_j$ contains at least $k+1$ subtrees, and Line \ref{lin:chooseProper} does not fail at $j$, a contradiction.
We conclude that $j$ must be the first node of $T$.
Therefore, the second modification is sufficient to guarantee that the chordal graph generated by the algorithm is $k$-connected. Clearly, these modifications do not change the time complexity of $\alg$.
\end{proof}

Noting that in a minimal representation $\rep$ of a chordal graph $G$, a subtree of $\TT$ containing an edge of $T$ yields an edge in $G$, Proposition \ref{prop:NumberOfTreesInEdge} implies the following simple characterization of minimal representations of an independent set. Indeed, this can be viewed as a reformulation of the following well known fact in terms of minimal representations: A chordal graph $G$ on $n$ vertices has at most $n$ maximal cliques, with equality if and only if $G$ is an independent set \cite{Golumbic:2004:AGT:984029}.

\begin{observation}\label{obs:IndependentSet}
Let $\rep$ be a minimal representation of a chordal graph $G$ on $n$ vertices where $\abs{V(T)}=t$. Then, $G$ is an independent set if and only if $t=n$.
\end{observation}

Proposition \ref{prop:NumberOfTreesInEdge} also gives some insight about minimal representations of special chordal graphs such as trees and forests as described in Propositions \ref{prop:tree} and \ref{prop:forest} respectively.

\begin{proposition}\label{prop:tree}
Let $\rep$ is a minimal representation of a connected chordal graph $G$ on $n \geq 2$ vertices where $\abs{V(T)}=t$. Then, $G$ is a tree if and only if $t=n-1$.
\end{proposition}
\begin{proof}
Assume that $t=n-1$. 
Since $G$ is connected, it remains to show that $G$ does not contain cycles.
Suppose that $G$ contains a cycle. 
Since $G$ is chordal, it must contain a triangle.
Therefore, the clique number of $G$ is at least 3.

Consider an execution of {\alg} that generates the representation $\rep$ and the graph $G=G(\TT)$.
Let $\TT^{(j)}$ be the set of subtrees of $\TT$ generated by the $j$-th invocation of {\func} (that created the node $j$ of $T$).
Then $\set{\TT^{(1)}, \ldots, \TT^{(n-1)}}$ is a partition of $\TT$.
Therefore, $\abs{\TT^{(j)}}=2$ for some $j \in [n-1]$ and $\abs{\TT^{(j')}}=1$, for every $j' \neq j$.
Corollary \ref{cor:AlgModification} implies that $j=1$.
By the behaviour of {\alg}, it is easy to show by induction on $j$, 
that, $\abs{\TT_j}=2$ for every $j \in [n-1]$.
That is, every maximal clique of $G$ has exactly two vertices, i.e. the clique number of $G$ is at most 2, a contradiction.

Conversely, assume that $G$ is a tree on $n$ vertices.
The set of maximal cliques of $G$ is the set of its $n-1$ edges. By Proposition \ref{prop:mincliquetree}, every minimal representation of $G$ has a tree $T$ on $t=n-1$ nodes.
\end{proof}

We note that a minimal representation of a tree is not unique. For example, let $G$ be a star on $n$ vertices. 
For every tree $T$ on $n-1$ vertices, $G$ has a representation $\rep$.
Indeed, $\TT$ consists of $T$ itself and the $n-1$ trivial subtrees of $T$.

If we relax the condition that $G$ is connected in Proposition \ref{prop:tree}, then we obtain the following result.

\begin{proposition}\label{prop:forest}
\begin{enumerate}[i)]
\item If $\rep$ is a minimal representation of a chordal graph $G$ on $n$ vertices, and $\abs{V(T)}=t=n-1$ then $G$ is a forest, where all the trees of $G$ except one are trivial.
\item Every forest on $n \geq 2$ vertices with exactly one non-trivial tree has a minimal representation $\rep$ with $\abs{V(T)}=t=n-1$.
\end{enumerate}
\end{proposition}

\begin{proof}
\begin{enumerate}[i)]
\item 
By Proposition \ref{prop:NumberOfTreesInEdge}, every edge of $T$ is contained in at most one subtree in $\TT$.
Let $T^{(1)}, \ldots, T^{(\ell)}$ be the subtrees of $T$ obtained by the removal of the edges $e$ such that $\TT_e = \emptyset$,
and $n_i$ be the number of subtrees in $\TT$ that are contained in $T^{(i)}$ for every $i \in [\ell]$.
Clearly, $\sum_{i=1}^\ell \abs{V(T^{(i)})} = t = n-1$, $\sum_{i=1}^\ell n_i = n$, 
and $n_i \geq \abs{V(T^{(i)})}$ for every $i \in \ell$.
Therefore, there exists some $j \in [\ell]$ such that $n_j = \abs{V(T^{(j)})}+1$,
and $n_{j'} = \abs{V(T^{(j')})}$ for every $j' \in [\ell]$ different than $j$. 
By Proposition \ref{prop:tree}, the subtrees in $T^{(j)}$ represent a tree.
By Observation \ref{obs:IndependentSet}, the subtrees in $T^{(j')}$ represent an independent set for every $j' \in [\ell]$ different than $j$.

\item
Let $G$ be a forest on $n$ vertices, with exactly one non-trivial tree $G'$ on $n'$ vertices.
Let $\repprime$ be a representation of $G'$ with $\abs{T'}=n'-1$.
By adding $n - n'$ new nodes to $T'$, 
and adding to $\TT'$ a trivial tree (consisting of that node) for every new node,
we get a representation of $G$ on a tree with $n-1$ nodes.
\end{enumerate}
\end{proof}

\section{Arbitrary Representations}\label{sec:general}
We start this section by showing that the upper bound of Theorem \ref{thm:MinimalUpperBound} does not hold for arbitrary representations on trees with $n$ nodes. 
We denote by $L(T)$ the set of leaves of a tree $T$.

\begin{lemma}\label{lem:lowerbound}
Let $\TT' = \set{T'_1, \ldots, T'_n}$ be a set of subtrees on a tree with $n$ nodes and $m$ be the number of edges of $G(\TT')$. Then
\[
\sum_{i=1}^n \abs{L(T'_i)} \textrm{~is~} \Omega (m \sqrt{n}).
\]
\end{lemma}
\begin{proof}
Let $k$ be a non-negative integer, and $n=6 \cdot 3^{2k}$.
Let $T'$ be a tree on $n$ nodes $\{v_1, \ldots, v_{4 \cdot 3^{2k}}, u_1, \ldots, u_{2 \cdot 3^{2k}}\}$, 
where the nodes $\{v_1, \ldots, v_{4 \cdot 3^{2k}}\}$ induce a path and the nodes $\{u_1, \ldots, u_{2 \cdot 3^{2k}}\}$ induce a star with center $u_1$.
The representation contains the following subtrees:
\begin{itemize}
\item {$S_1$:} 2 trivial paths on every node in $\{v_1, \ldots, v_{2 \cdot 3^{2k}}\}$, for a total of $4 \cdot 3^{2k}$ paths,
\item {$S_2$:} $3^{k+1}$ copies of the star on nodes $\{u_1, \ldots, u_{2 \cdot 3^{2k}}\}$, and
\item {$S_3$:} $2 \cdot 3^{2k} - 3^{k+1}$ disjoint trivial paths on part of the nodes in $\{v_{2 \cdot 3^{2k} + 1}, \ldots, v_{4 \cdot 3^{2k}}\}$. 
Note that the number of these paths is less than the number $2 \cdot 3^{2k}$ of nodes in the path, thus disjointness can be achieved. 
\end{itemize}
The number of subtrees is $4 \cdot 3^{2k} +  3^{k+1} + 2 \cdot 3^{2k} - 3^{k+1} = 6 \cdot 3^{2k} = n$, as required. 
As for the total number of leaves, we have:
\[
 \sum_{i=1}^n \abs{L(T'_i)} \geq 3^{k+1} \cdot (2 \cdot 3^{2k}-1) + 4 \cdot 3^{2k} + 2 \cdot 3^{2k} - 3^{k+1} \geq 6 \cdot 3^{3k}.
\]
Let $G$ be the intersection graph of these subtrees. $G$ consists of a $K_{3^{k+1}}$, $2 \cdot 3^{2k}$ disjoint $K_2$s and isolated vertices.
We have $m =  2 \cdot 3^{2k} + 3^{k+1} \frac{3^{k+1} - 1}{2}$, thus
\[
\frac{\sum_{i=1}^n \abs{L(T'_i)}}{m} \geq \frac{6 \cdot 3^{3k}}{\frac{13}{2} \cdot 3^{2k} - \frac{3}{2}3^k} = \Omega(3^k) = \Omega(\sqrt{n}).
\]
\end{proof}

Since the space needed to represent a subtree is at least the number of its leaves, Lemma \ref{lem:lowerbound} implies the following:

\begin{corollary}
The time complexity of any algorithm that generates chordal graphs by picking an arbitrary subtree representation on a tree with $n$ nodes is $\Omega(m \sqrt{n})$.
\end{corollary}
We now proceed to show that this bound is tight up to a constant factor.

Through the rest of this section $T'$ is a tree on $n$ nodes and $\TT'=\set{T'_1, \ldots, T'_n}$ is a set of subtrees of $T'$ such that $G(\TT')=G$. We also denote by $t'_j \defined \abs{\TT'_j}$ the number of subtrees in $\TT'$ that contain the node $j$ of $T'$. We assume that $\rep$ is a contraction-minimal representation of $G$ obtained from $\repprime$ by zero or more successive contraction operations. Then, $T$ has $t\leq n$ nodes and the \emph{multiplicity} of $j$ (with respect to $T'$), denoted by $k_j$, is the number of contractions effectuated in $T'$ in order to obtain node $j$ in $T$, plus one.

\begin{lemma}
With the above notations, we have 
\begin{eqnarray}
\forall j \in [t]~~k_j \geq 1 \label{eqn:kPositive}\\
\sum_{j=1}^{t} k_j = n.\label{eqn:kSumsToN}\\
\sum_{i=1}^n \abs{V(T'_i)} \leq \sum_{j=1}^{t} k_j t_j. \label{eqn:upperb}
\end{eqnarray}
\end{lemma}

\begin{proof}
Relations (\ref{eqn:kPositive}) and (\ref{eqn:kSumsToN}) hold by definition of $t$ and multiplicity.

When an edge $jj'$ of $T'$ is contracted to a node $v$ we have $k_v = k_j + k_{j'}$ and $t_v = \max \set{t_j,t_{j'}}$.
Therefore,
\[
k_v t_v = (k_j + k_{j'}) t_v \geq k_j t_j + k_{j'} t_{j'}.
\] 
Using the above fact and noting that we have $\sum_{i=1}^n \abs{V(T'_i)} = \sum_{j=1}^n t'_j$, inequality (\ref{eqn:upperb}) follows by induction on the number of contractions.
\end{proof}
Recall that our task is to find an upper bound on the sum of the sizes of subtrees in a representation on a tree with $n$ nodes of a chordal graph on $n$ vertices. 
Relation (\ref{eqn:upperb}) allows us to focus on the sum of $k_jt_j$ values (in a minimal representation) in order to achieve this goal. 
In what follows, we treat this task as an optimization problem under a given set of constraints. 
Thus, the following lemma can be (and should be) read independently from graph theoretic interpretations of each parameter. 

\begin{lemma}\label{lem:optimization}
Let $t$, $s_1, \ldots, s_t$, $t_1, \ldots t_t$ and $k_1, \ldots, k_t$ be numbers that satisfy 
(\ref{eqn:sPositive}), (\ref{eqn:sSumsToN}), (\ref{eqn:boundsForT}), (\ref{eqn:graphSize}), (\ref{eqn:kPositive}), and (\ref{eqn:kSumsToN}). Then
\[
\rho \defined \frac{\sum_{j=1}^{t} k_j t_j}{2m + n} \textrm{~is~} {\cal O} (\sqrt{n}).
\]
\end{lemma}
\begin{proof}
Let $k_j, s_j, t_j (j \in [t])$ be values that maximize $\rho$ which we refer as \emph{optimal}.
Such values exist since $k_j, t_j \leq n$, thus $\rho \leq \frac{n^3}{2m+n}$.
Let also $\lambda_j = k_j / s_j$.
If $t=1$ we have $k_1=s_1=t_1=n$. Then (\ref{eqn:graphSize}) implies $2m+n=n^2$, thus $\rho=1$.
Therefore, in the rest of the proof we assume $t \geq 2$.

Since the only constraints for $k_j$ are \eqref{eqn:kPositive} and \eqref{eqn:kSumsToN}, and $k_j$ does not appear in the denominator of $\rho$, 
there are optimal values where
\[
k_j = \left\{
\begin{array}{ll}
n-t+1 & \textrm{~if~} j=t\\
1 & \textrm{otherwise}.
\end{array}
\right.
\]
Therefore, we have $\lambda_j = \frac{1}{s_j} \leq 1$ for all $j \in [t-1]$.
This implies $\lambda_{t} \geq 1$ since otherwise $n= \sum_{j=1}^t k_j = \sum_{j=1}^{t-1} k_j + k_t \leq \sum_{j=1}^{t-1} s_j + k_t < \sum_{j=1}^{t-1} s_j + s_t = n$, a contradiction.

By \eqref{eqn:boundsForT} we have $s_t \leq t_t=\sum_{i=t}^t s_i=s_t$, thus $t_t=s_t$.
Furthermore, $t_j<\sum_{i=j}^t s_i$ for every $j<t$ since otherwise $t_j = \sum_{i=j}^t s_i > \sum_{i=t}^t s_i = t_{t}$, contradicting the maximality of $t_{t}$ (recall that $t_t=\max\{t_j|j\in [t]\}$). 
Also, $s_{t}=t_{t} \geq t_j \geq s_j$ for every $j \in [t]$. 
Let $\ell \in [t-1]$ be such that $s_\ell = \min_{j\in [t-1]} s_j$.

The following claim guarantees optimal values with an even more restricted structure.
\begin{claim}\label{clm:MostTAreExtremal}
There are optimal values with $s_j=t_j$ for every $j \in [t]$ different from $\ell$.
\end{claim}
\begin{proof}
Consider optimal values minimizing $\sum_{j \neq \ell} t_j$ among all optimal values that satisfy relations (\ref{eqn:sPositive}) - (\ref{eqn:kSumsToN}) and where $k_j$s are defined as above. Assume by way of contradiction that these values do not satisfy the claim, i.e. there exists $j \neq \ell$ with $s_j < t_j$. 
Since $t_{t}=s_{t}$, we have $j \neq t$.
By the choice of $\ell$ we have $\lambda_j \leq \lambda_\ell$, and recall that $t_{\ell} < \sum_{i=\ell}^t s_i$.

Let $t'_\ell = t_\ell + \epsilon$ and $t'_j = t_j - \epsilon \frac{s_\ell}{s_j}$ for some $\epsilon > 0$. We have
\[
s_j t'_j + s_\ell t'_\ell = s_j (t_j - \epsilon \frac{s_\ell}{s_j}) + s_\ell (t_\ell + \epsilon) = s_j t_j + s_\ell t_\ell.
\]
Therefore, if we replace $t_j, t_\ell$ by $t'_j$ and $t'_\ell$ respectively, relation (\ref{eqn:graphSize}) remains valid. Moreover, we can choose $\epsilon$ sufficiently small so that all constraints \eqref{eqn:boundsForT} are still satisfied. 
As for the numerator of $\rho$, we have
\[
k_j t'_j + k_\ell t'_\ell = k_j (t_j - \epsilon \frac{s_\ell}{s_j}) + k_\ell (t_\ell + \epsilon) = k_j t_j + k_\ell t_\ell + \epsilon s_\ell (\lambda_\ell - \lambda_j ) \geq k_j t_j + k_\ell t_\ell,
\]
i.e. the values are optimal. 
However $\sum_{j \neq \ell} t_j$ decreased by $\epsilon \frac{s_\ell}{s_j}$, contradicting the way the optimal values were chosen.
\end{proof}
We now note that $t_\ell \leq t_t = s_t = n - \sum_{j=1}^{t-1}s_j$.
Therefore, we have
\begin{eqnarray}
\sum_{j=1}^{t} k_j t_j & = & \sum_{j \in [t-1] \setminus \set{\ell} } s_j + (n-t+1) s_t + t_\ell 
\leq \sum_{j \in [t-1] \setminus \set{\ell} } s_j + (n-t+1) s_t + n - \sum_{j=1}^{t-1}s_j\nonumber \\
& < & n+(n-t+1)s_t \label{eqn:numerator}\\
2m+n & = & 2 \sum_{j=1}^{t} s_j t_j - \sum_{j=1}^{t} s_j^2 \geq \sum_{j=1}^{t} s_j^2 = \sum_{j=1}^{t-1} s_j^2 + s_t^2 \geq \frac{(n-s_t)^2}{t-1}+s_t^2\label{eqn:denominator}
\end{eqnarray}

where the last inequality holds because $\sum_{j=1}^{t-1} s_j = n - s_{t}$ and when the sum is fixed, the sum of squares is smallest when all numbers are equal. Now, we note that $0 < 2(n-t)+1$ and this is equivalent to $(t-1)n < t(2n-t)-(t-1)(n-t+1)$ which implies in turn the following inequality since $s_t\geq 1$:
\begin{eqnarray}
(t-1) (n+(n-t+1)s_t) < t (2n-t) s_t. \label{eqn:ineq}
\end{eqnarray}

We combine inequalities (\ref{eqn:numerator}), (\ref{eqn:denominator}) and (\ref{eqn:ineq}) as follows
\begin{eqnarray*}
\rho < \frac{n + (n-t+1) s_t}{\frac{(n-s_t)^2}{t-1} + s_t^2} =
              \frac{(t-1) (n + (n-t+1) s_t)}{(n-s_t)^2 + (t-1) s_t^2}
     < \frac{t (2 n - t) s_t}{n^2 -2 n s_t + t s_t^2} = \bar{\rho}(s_t).
\end{eqnarray*}
$\frac{\partial \bar{\rho}}{\partial s_{t}} \geq 0$ if and only if
$n^2 - 2 n s_{t} + t s_{t}^2 \geq s_{t} (2 t s_{t} - 2 n)$ if and only if
$n^2 \geq t s_{t}^2$.
Therefore, $\bar{\rho}$ is maximum when $s_{t}=\frac{n}{\sqrt{t}}$. 
We conclude
\begin{eqnarray*}
\rho & < & \bar{\rho}(\frac{n}{\sqrt{t}}) = \frac{t (2 n - t) n / \sqrt{t}}{n^2 -2 n^2 / \sqrt{t} + n^2} =
\frac{t (2 n - t)}{2 n (\sqrt{t} -1)} \leq \frac{2t (2 n - t)}{n \sqrt{t}} = \frac{2\sqrt{t} (2 n - t)}{n} = \bar {\bar \rho} (t)
\end{eqnarray*}
where we use the fact that $t \geq 2$ to derive the inequality.
We substitute $u=\sqrt{t}$ to get
\[
\bar {\bar \rho} (t) = \frac{2u (2n - u^2)}{n} = 4 u - \frac{2u^3}{n}
\]
and we derive
\[
\frac{\partial {\bar {\bar \rho}}} {\partial u} = 4 - \frac{6 u ^2}{n}.
\]
Since $\frac {\partial u}{\partial t} > 0$, $\bar {\bar \rho}$ attains its maximum at $t=u^2=\frac{2n}{3}$. We conclude
\[
\bar {\bar \rho} (t) \leq \bar {\bar \rho} (2n/3) = 4 \sqrt{\frac{2n}{3}} - \frac{2}{n} \cdot \frac{2n}{3} \sqrt{\frac{2n}{3}}=\frac{8}{3}\sqrt{\frac{2n}{3}}.
\]

\end{proof}

We can now infer the following theorem.
\begin{theorem}\label{lem:UpperBound}
Let $\repprime$ be a representation of a chordal graph $G$ where $T'$ has $n$  nodes and $G$ has $n$ vertices and $m$ edges. Then we have
\[
\sum_{i=1}^n \abs{V(T'_i)}  \textrm{~is~} \Theta (m \sqrt{n}).
\]
\end{theorem}
\begin{proof}
Let $\rep$ be a minimal representation of $G$ obtained from $\repprime$. Then $\rep$ satisfies (\ref{eqn:sPositive}), (\ref{eqn:sSumsToN}), (\ref{eqn:boundsForT}), (\ref{eqn:graphSize}) by Lemma \ref{lem:minrep}, and (\ref{eqn:kPositive}),  (\ref{eqn:kSumsToN}), (\ref{eqn:upperb}) hold by Lemma \ref{lem:UpperBound}. The lower and upper bounds provided in Lemmas \ref{lem:lowerbound} and \ref{lem:optimization} respectively allows us to conclude the proof.
\end{proof}

\section{Conclusion}\label{sec:conclusion}
In this work, we present a linear time algorithm to generate random chordal graphs. 
To the best of our knowledge, this is the first algorithm with this time complexity.
Our algorithm is fast in practice and simple to implement. 
We also show that the complexity of any random chordal graph generator which produces any subtree intersection representation on a tree of $n$ nodes with positive probability is $\Omega(m\sqrt{n})$.

We conducted experiments to analyze the distribution of the sizes of the maximal cliques of the generated chordal graphs. As a result, we concluded that our method generates fairly varied chordal graphs with respect to this measure. It should be noted that, however, we do not know the distribution of the maximal clique sizes over the space of all chordal graphs of a given size.

We have shown that every chordal graph on $n$ vertices is returned by our algorithm with positive probability. 
Our algorithm allows us to analyze connectivity properties of the generated chordal graphs, and can be used to generate chordal graphs having a given (vertex) connectivity.

The development of an algorithm that generates chordal graphs uniformly at random is subject of further research.

%


\end{document}